\documentclass[11pt,letterpaper]{article}
\pdfoutput=1

\usepackage{microtype}%if unwanted, comment out or use option "draft"

\usepackage[T1]{fontenc}
\usepackage{amsthm}
\usepackage{amssymb}
\usepackage[english]{babel}
\usepackage[utf8]{inputenc}
\usepackage{amsmath}
\usepackage{amsfonts}
\usepackage{graphicx}
\usepackage[colorinlistoftodos]{todonotes}
\usepackage{mathtools}
\usepackage{enumitem}
\usepackage{thmtools}
\usepackage{fullpage}
\usepackage{booktabs}
\usepackage{bbm}
\usepackage[normalem]{ulem} % m.in. przekreślenie tekstu
\usepackage[colorlinks,breaklinks,linkcolor=teal, citecolor = teal, urlcolor=gray]{hyperref}

\newcommand{\bigo}{\mathcal{O}}

\newcommand{\N}{\mathbb{N}}

\newcommand{\G}{{\mathcal G}}

\newcommand{\Gnp}{\G(n,p)}
\newcommand{\R}{\mathbb{R}}
\newcommand{\E}{\mathbb E}
\newcommand{\Prob}{\mathbb{P}}

\theoremstyle{plain}
\newtheorem{definition}{Definition}[section]
\newtheorem{lemma}[definition]{Lemma}
\newtheorem{theorem}[definition]{Theorem}

\newtheorem{proposition}[definition]{Proposition}
\newtheorem{property}[definition]{Property}

\title{Asynchronous Majority Dynamics on Binomial Random Graphs\thanks{This work was done while the authors were visiting the Simons Institute for the Theory of Computing. This project was funded in part by the European Research Council (ERC) under the European Union's Horizon 2020 research and innovation program (grant agreement no.\ 866132) and by the Israel Science Foundation (grant no.\ 317/17), when D.\ Mohan was at Tel Aviv University. A short version of the paper appeared at APPROX 2024.}
}

\author{
Divyarthi Mohan\thanks{Faculty of Computing \& Data Sciences, Boston University, USA, email: \texttt{dmohan@bu.edu}.}
\and
Paweł Prałat\thanks{Department of Mathematics, Toronto Metropolitan University, Toronto, ON, Canada, email: \texttt{pralat@torontomu.ca}.}
}

\date{}

\begin{document}

\maketitle

\begin{abstract}
We study information aggregation in networks when agents interact to learn a binary state of the world. Initially each agent privately observes an independent signal which is \emph{correct} with probability $\frac{1}{2}+\delta$ for some $\delta > 0$. At each round, a node is selected uniformly at random to update their public opinion to match the majority of their neighbours (breaking ties in favour of their initial private signal). Our main result shows that for sparse and connected binomial random graphs $\Gnp$ the process stabilizes in a \emph{correct} consensus in $\bigo(n\log^2 n/\log\log n)$ steps with high probability. In fact, when $\log n/n \ll p = o(1)$ the process terminates at time $\hat T = (1+o(1))n\log n$, where $\hat T$ is the first time when all nodes have been selected at least once.  However, in dense binomial random graphs with $p=\Omega(1)$, there is an information cascade where the process terminates in the \emph{incorrect} consensus with probability bounded away from zero.
\end{abstract}

%%%%%%%%%%%%%%%%%%%%%%%%%%%%%%%%
\section{Introduction}
%%%%%%%%%%%%%%%%%%%%%%%%%%%%%%%%

Our opinions and actions we take as individuals are often influenced by both our private knowledge of the world and the information we obtain through our interactions with others. For example, a voter deciding which candidate's economics policies would decrease inflation, might have an initial belief based on her own past expenditure and later might be swayed by her friends' opinions. Now more than ever, with the advent of social media and online platforms, our interactions have increased many folds and our social networks are massive. Hence, an important research question is to understand if and how the structure of the social network and the dynamics of the interactions impact the (mis)information propagated~\cite{matwin2023generative}. Do our social networks enable successful information aggregation and lead to social learning, or do they amplify incorrect beliefs leading to an information cascade? 

There has been extensive work modeling these opinion dynamics formally to study the network effects on information aggregation; see Section~\ref{sec:related}. In this paper, we focus on the model of \emph{asynchronous majority dynamics}, where agents in a network (asynchronously) update their opinions to match the majority opinion amongst their neighbours. In particular, each agent initially has a private belief over a binary state of the world and no publicly announced opinion. At each time step, an agent is chosen uniformly at random to announce/update her opinion and she does so by simply copying the majority of the neighbours' current announced opinions, breaking ties with her initial belief. Majority dynamics is clearly a na\"ive learning model, as the agents do not reason about potential information redundancy due to interaction between one's neighbours. Such na\"ive learning (non-Bayesian) models are a more faithful abstraction of everyday interactions between agents with bounded rationality (e.g., voters and consumers), while Bayesian models are a better abstraction of rational agents or interactions about high-stakes information (e.g., traders and scientists). We consider asynchronous updates which are more suitable to capture human decision making. Moreover, asynchronous emergence of announcements also captures an initial information diffusion phase before conventional social learning starts.

In our model, there is a \emph{correct} opinion (i.e., the true state of the world) and each agent's initial private belief is independently drawn and is biased towards being correct (with probability $1/2 + \delta$). So initially, in a large network, there is enough information so that an omniscient central planner can infer the true state (with very high probability). However, agents in the network are updating their opinions based on local heuristics, so the network structure can crucially alter the final outcome of the dynamics. For example, in a complete graph, with a constant probability all the nodes converge to the wrong opinion. On the other hand, in a star graph with high probability all the nodes converge to the correct opinion.  This brings us to the main question of interest:
\begin{quote}
    \emph{``What network structures enable efficient social learning, where the dynamics stabilizes with every agent in the network reaching the correct opinion?''}
\end{quote}

Feldman et al.~\cite{FeldmanILW14}, who initiated the study of asynchronous majority dynamics, showed that when the network is sparse (has bounded degree) and expansive, a correct consensus is reached with high probability. More recently, Bahrani et al.~\cite{BahraniIMW20} studied networks that have certain tree structures (like preferential attachment trees and balanced $m$-ary trees) and showed that the dynamics stabilizes in a correct majority. Both results heavily rely on these particular assumptions on the network. For example, to even establish that a \emph{majority of the nodes} have the correct opinion \emph{at some point} in the process, it is crucial that the network is either a bounded degree graph or is a tree. In this paper, our goal is to extend the guarantees of asynchronous majority dynamics beyond these assumptions and to develop techniques applicable to more general networks formed through random graph models.  

%%%%%%%%%%%%%%%%%%%%%%%%%%%%%%%%
\subsection{The Model}\label{sec:model}

Consider any undirected graph $G=(V,E)$ on $n=|V|$ nodes. Individuals initially have one of two private \emph{beliefs} which we will refer to as ``Correct'' (or 1) and ``Incorrect'' (or 0). Formally, each $v \in V(G)$ receives an independent private signal $X(v) \in \{0, 1\}$, and $\Pr( X(v) = 1 ) = 1/2 + \delta$, for some universal constant $\delta \in (0,1/2)$. Individuals also have a publicly announced \emph{opinion} which we will simply refer to as an announcement or opinion. We define $C^t(v) \in \{\perp, 0, 1\}$ to be the public announcement of $v \in V$ at time $t$. 

Initially, no announcement have been made, that is, $C^0(v) = \perp$ for all $v \in V$. In each subsequent step, a single node $v^t$ is chosen uniformly at random from $V$, independently from the history of the process. In particular, as in the classical coupon collector problem, some nodes will be chosen many times before others will get lucky to get chosen for the first time. In step $t$, $v^t$ updates her announcement using \emph{majority dynamics}, while announcements of other nodes stay the same. To be specific, for any $i \in \{\perp, 0, 1\}$ and $v \in V$, let $N^t_i(v)$ denotes the number of neighbours of $v$ that have opinion $i$ at time $t$. Then,
$$
C^t(v) = 
\begin{cases}
1 & \text{ if } N_1^{t-1}(v) > N_0^{t-1}(v) \text{ and } v=v^t,\\
0 & \text{ if } N_1^{t-1}(v) < N_0^{t-1}(v) \text{ and } v=v^t,\\
X(v) & \text{ if } N_1^{t-1}(v) = N_0^{t-1}(v) \text{ and } v=v^t,\\
C^{t-1}(v) & \text{ if } v \neq v^t.
\end{cases}
$$
That is, at time step $t$, the opinion of the chosen node $v^t$ is updated to match the majority opinion among the current public announcements of $v^t$'s neighbours and ties are broken in favour of $v^t$'s initial private belief. Observe that when none of $v^t$'s neighbours have announced so far, i.e., $N_1^{t-1}(v)=N_0^{t-1}(v)=0$, we have $C^t(v)=X(v)$.

Finally, for any $i \in \{\perp, 0, 1\}$, let $Y_i^t$ be the number of nodes that have opinion $i$ at time $t$, that is, $Y_i^t = | \{ v \in V : C^t(v) = i \} |$.

As shown in~\cite{FeldmanILW14}, it is easy to see that in any network this process stabilizes with high probability in $\bigo(n^2)$ steps. In fact, the process stabilizes in $O(n\log n + n\cdot d(G))$ where $d(G)$ is the diameter of the graph~\cite{BahraniIMW20}. That is, the network reaches a state at some time $T$ where no node will want to change its announcement and thus the process terminates. Our goal is to understand what fraction of nodes converges to the correct opinion, that is, what the value of $Y_1^T / n$ is.

%%%%%%%%%%%%%%%%%%%%%%%%%%%%%%%%
\subsection{Our Results}\label{sec:results}

The main contribution of this paper is the proof that the asynchronous majority dynamics on \emph{binomial random graph} $\G(n,p)$ converges to the correct opinion, provided that the graph is sparse (that is, the average degree $np = o(n)$) and connected (that is, $np - \log n \gg 1$). If $np \gg \log n$, then the process converges to the correct opinion as quickly as it potentially could. 

\begin{theorem}\label{thm:not_very_sparse_graphs}
Let $\delta \in (0, 1/10]$. Let $\omega'=\omega'(n)=o(\log n)$ be any function that tends to infinity as $n \to \infty$. Suppose that $p = p(n) \ll 1$ and $p \gg \log n/n$, and consider the asynchronous majority dynamics on $\G(n,p)$. 

Then, asymptotically almost surely (a.a.s.) after $n(\log n + \omega') = (1+o(1)) n \log n$ rounds the process terminates with all nodes announcing the correct opinion. In fact, it happens exactly at time $\hat{T}$, where $\hat{T}$ is the first time when all nodes are selected at least once.
\end{theorem}

For sparser (but still connected) graphs, the process also converges to the correct opinion. In this case, we do not aim to show that it happens at time $\hat{T}$ and we only provide an upper bound for the number of rounds. It remains an open problem to determine if the process terminates at time $\hat{T}$ or it needs more time to converge.

\begin{theorem}\label{thm:very_sparse_graphs}
Let $\delta \in (0, 1/10]$. Let $\omega'=\omega'(n)=o(\log \log n)$ be any function that tends to infinity as $n \to \infty$. Suppose that $p = p(n) \le \omega' \log n / n$ and $p \ge (\log n+\omega')/n$, and consider the asynchronous majority dynamics on $\G(n,p)$. 

Then, a.a.s.\ after $\bigo (n (\log n)^2 / (\log \log n))$ rounds the process terminates with all nodes announcing the correct opinion. 
\end{theorem}

These results are best possible in the following sense. If $p \le (\log n-\omega')/n$, then a.a.s.\ $\G(n,p)$ is disconnected. In fact, a.a.s.\ there are at least $\omega'$ isolated nodes which announce their own private believes. As a result, a.a.s.\ some nodes announce the correct opinion but some of them announce the incorrect one. Indeed, the probability that all isolated nodes converge to the same opinion is at most $o(1) + (1/2+\delta/2)^{\omega'} + (1/2-\delta/2)^{\omega'} = o(1)$. On the other hand, if $p \in (0,1]$ is a constant separated from zero, then with positive probability the process converges to the correct opinion and with positive probability it converges to the incorrect opinion. 

\begin{theorem}\label{thm:dense_graphs}
Let $\delta \in (0, 1/2)$. Let $\omega'=\omega'(n)=o(\log n)$ be any function that tends to infinity as $n \to \infty$. Suppose that $p \in (0,1]$ is a constant, and consider the asynchronous majority dynamics on $\G(n,p)$. 

Then, the following is true for $i \in \{0,1\}$: with probability at least $p_i$, after $n(\log n + \omega') = (1+o(1)) n \log n$ rounds the process terminates with all nodes announcing opinion~$i$, where
\begin{eqnarray*}
p_1 &=& (1/2+\delta) \exp \Big( - \log (1/p) (1/p) \Big) =(1/2+\delta)p^{1/p} > 0\\
p_0 &=& (1/2-\delta) \exp \Big( - \log (1/p) (1/p) \Big) =(1/2-\delta)p^{1/p} > 0.
\end{eqnarray*}
\end{theorem}

Let us stress the fact that the constants $p_0$ and $p_1$ in Theorem~\ref{thm:dense_graphs} depend on $\delta$ and $p$. Both dependencies are necessary. In particular, Theorem~\ref{thm:not_very_sparse_graphs} implies that $p_0 \to 0$ as $p \to 0$ so there is no universal constant that works for \emph{all} values of $p \in (0,1]$. Finally, let us mention that for some technical reason, in Theorems~\ref{thm:not_very_sparse_graphs} and~\ref{thm:very_sparse_graphs} it is assumed that $\delta \le 1/10$. However, it is easy to couple the process with $\delta \le 1/10$ with the one with $\delta \in (1/10, 1/2)$ to show that the result holds for any $\delta \in (0, 1/2)$---see Subsection~\ref{sec:coupling} for more details.

\paragraph{Sparsity and Social Learning.}  In summary, we show that random graphs with sub-linear average degree converges to the correct consensus a.a.s. On the other hand, random graphs with linear average degree may converge to an incorrect consensus with probability bounded away from zero. These results suggest that network sparsity is fundamental for successful social learning: communities with \emph{sparse} connections are very likely to \emph{collectively learn} the correct opinion, while \emph{dense} connections may enable \emph{information cascades} where the community converging to the incorrect opinion while ignoring the correct private beliefs held by a majority.

Moreover, our analysis highlights a key phenomena behind this difference. Sparsely connected networks allow various private beliefs to emerge independently in the network initially, which enables the correct opinion to win eventually. In contrast, densely connected networks promote rapid copying of early opinions, which overwhelms the natural advantage of the majority correct private beliefs. Further, we find that the level of sparsity can determine how quickly the correct consensus is reached. In particular, with sufficient connections, the correct opinion propagates to nearly the entire network faster than the coupon collector bound (on time $\hat T$ by which all nodes have announced at least once). Consequently, when the network is sparse but not too sparse (i.e., with average degree $\omega(\log n)$ but still sub-linear), the process converges to consensus precisely when the last node makes its first announcement.

These findings strengthens the results of~\cite{FeldmanILW14}, who show that a more restrictive notion of sparsity---{constant maximum degree}---facilitates social learning in expander graphs. While our results focus on a theoretical model of asynchronous majority dynamics, it complements experimental results of Dasaratha and He~\cite{DasarathaH21}, observe that people learn better in sparser binomial random graphs.

%%%%%%%%%%%%%%%%%%%%%%%%%%%%%%%%
\subsection{Future Directions}

Let us highlight a few potential directions one might want to consider.

\begin{itemize}
\item As already mentioned above, for very sparse graphs ($np - \log n \to \infty$ and $np = \bigo(\log n)$), it would be interesting to determine if the process terminates at time $\hat{T}$ or it needs more time to converge to the correct opinion---see Theorem~\ref{thm:very_sparse_graphs}.
\item Theorem~\ref{thm:very_sparse_graphs} holds as long as $pn = \log n + \omega$ for some $\omega = \omega(n) \to \infty$ as $n \to \infty$. It is known that if $pn = \log n + c$ for some constant $c \in \R$, then with probability bounded away from one and from zero the graphs is disconnected. As a result, there is no hope to extend the result for this range of $p$. But it is plausible that a.a.s.\ it holds right at the time the random graph process creates a connected graph. This would be an optimal ``hitting time'' result. 
\item For disconnected graphs ($np - \log n \to -\infty$), it would be interesting to investigate the process run on the giant component of $\G(n,p)$.
\item For dense graphs, it is not true that a.a.s.\ all nodes converge to the correct opinion---see Theorem~\ref{thm:dense_graphs}. Having said that, it is reasonable to expect that a.a.s.\ all nodes converge to the same opinion (for example, \cite{FountoulakisKM20} show that a consensus is reached in this case in a synchronous setting). Is is true in our asynchronous setting? In any case, what is the asymptotic value of the probability that all nodes converge to the correct opinion?
\item It would be interesting to investigate other random graph models that are able to generate graphs with power-law degree distributions as the Chung-Lu model~\cite{chung2006complex} or the classical configuration model. More challenging, but an important and interesting, direction would be to understand the learning process on a network with a community structure such as the ABCD (Artificial Benchmark for Community Detection) model~\cite{kaminski2021artificial} which produces a random graph with community structure and power-law distribution for both degrees and community sizes. In this model, small communities might create echo chambers, environments in which participants encounter beliefs that amplify or reinforce their preexisting beliefs inside a community and insulated from rebuttal.
\item In this paper, we analyzed the model introduced earlier~\cite{FeldmanILW14} that already generated some interests. However, it is reasonable to assume that agents (in lab experiments) are likely to keep their private signals even when the difference between public opinions among their neighbours is small. Hence, it would be interesting to analyze the model in which agents announce their private beliefs in such situation.
\end{itemize}

%%%%%%%%%%%%%%%%%%%%%%%%%%%%%%%%
\subsection{Related Work}\label{sec:related}

In this section, we briefly discuss prior work on social learning mainly focusing on the setting with a binary state of the world and the agents initially have a correct opinion independently with probability $1/2+\delta$. We refer to some recent surveys on social learning and opinion dynamics~\cite{MobiusR14,BecchettiCN20,BikhchandaniHTW21} for a more detailed literature review. 

Majority dynamics falls under a wide class of naive or non-Bayesian models, where agents use a simple local heuristic to update their opinions, to capture simple behaviours exhibited by non-expert decision makers. Prior works have studied majority dynamics under a variety of modeling assumptions, to understand when a consensus is possible and when there is social learning---that is, the consensus (or the majority) is correct. These works study a variety of networks such as $k$-regular trees~\cite{Howard00,KanoriaM11a,BahraniIMW20}, bounded degree graphs~\cite{FeldmanILW14}, random regular graphs~\cite{GartnerZ2018majority}, ``symmetric'' graphs and expanders~\cite{MosselNT13}. In~\cite{TamuzT13}, a different perspective on social learning asks when is it possible to ``recover the correct opinion'' at the end of the dynamics through any function (not just a consensus or majority vote). Prior work has also considered models with different notions of bias towards correct opinion, for example, each node updates to the correct opinion with some probability~\cite{AnagnostopoulosBCPR22}, or the initial configuration of the network has some $n/2 + \delta$ correct opinions~\cite{TranV19,tran2023power}. 

Recently, there has been a series of work studying synchronous majority dynamics in binomial random graphs~\cite{BenjaminiCOTT16,FountoulakisKM20,ChakrabortiHLT23}, with a focus to showing that $99\%$ of the nodes converge to the same opinion (with high probability) for sparse random graphs, with $p = \Omega(\log n/n^{3/5})$ being the best known {lower bound for the average degree}. Moreover,~\cite{zehmakan2020opinion} showed that a correct consensus is reached with high probability for binomial random graphs with $p=\Omega(\log n / n)$ under synchronous majority dynamics. In contrast to these works, we focus on asynchronous dynamics and prove that a correct consensus is reached with high probability for {$p = \Omega(\log n / n)$ and $p=o(1)$}. Binomial random graphs are also studied under label propagation~\cite{kiwi2023label} which is a special case of synchronous majority dynamics with non-binary opinion in $[0,1]$. 

Many of the works mentioned above focus on synchronous updates, where all agents update their opinions synchronously in each round. Majority dynamics with synchronous updates leads to a correct consensus for all networks that are sufficiently connected~\cite{MosselNT13}, whereas with asynchronous updates the network structure can have a huge impact on social learning. This is best illustrated by the complete graph. With asynchronous updates, once the first agent announces their opinion (which can be wrong with probability $1/2 - \delta$) everyone will copy this. Hence, with {a probability bounded away from zero} all the nodes converge to the wrong opinion. In contrast, if all agents were to update synchronously, then the majority of the round one updates will be correct  with high probability, so there will be a correct consensus in round two. Recent work~\cite{BanerjeeBCM21}, studies the DeGroot model with uninformed agents, to capture the different phases of information diffusion and social learning, which is a key phenomena that occurs in our asynchronous model. 

Asynchronous dynamics, where a random node is chosen to update at each time, have also been studied under different modeling assumptions. In~\cite{SchoenebeckY18}, there are no private beliefs, instead initially all nodes have a some publicly announced opinions and ties are broken at random. They show that for any initial configuration a consensus is reached with high probability in time $O(n\log n)$ in dense binomial graphs ($p=\Omega(1)$) under a general class of majority-like update rules.  They leave it as an open problem to study the consensus time of sparse binomial random graph.

Other non-Bayesian dynamics have also been extensively studied. In the Voter model, agents choose a random neighbour and copy their opinion~\cite{Cliffords73,HolleyL75}. A similar dynamics called $k$-majority model are studied in the distributed computing literature, where agents choose $k$-neighbours at random and copy their majority~\cite{BechettiCNPL16,GhaffariP16,GhaffariL18,CrucianiMQR21, AbdullahD2015global}. In the DeGroot Model, an agent's opinion lies in $[0,1]$ (as opposed to binary $\{0,1\}$) and agents update to the average of their neighbours~\cite{DeGroot74,GolubJ10}. {In~\cite{ElboimPP24}, an asynchronous DeGroot dynamics is considered where each node has an independent Poisson clock which determines when they are chosen to update.} A key difference between these works and majority dynamics is that in these models a consensus is reached with probability $1$ for any connected graphs. This is not the case in majority dynamics even with synchronous updates. 

While our focus is in non-Bayesian dynamics, there has also been a long line of work studying Bayesian models, where agents update their beliefs rationally given their (local) observations exhibiting more sophisticated decision-making. Seminal works~\cite{Banerjee92,BikchandaniHW92} introduced the study of Bayesian dynamics and identified conditions that lead to information cascades. Here, the agents arrive sequentially and observe all the announcements (i.e., they form a complete graph), and many other subsequent works consider Bayesian dynamics under different assumptions and variations~\cite{SmithS00,BanerjeeF04,CelenK04}. Bayesian dynamics in general social networks were first studied in~\cite{AcemogluDLO11}. There is also a long line of work studying Bayesian learning with repeated interactions~\cite{GaleK03,RosenbergSV09,KanoriaT13,MullerF13,MosselST14,MosselOT16}. 

%%%%%%%%%%%%%%%%%%%%%%%%%%%%%%%%
\section{Preliminaries}
%%%%%%%%%%%%%%%%%%%%%%%%%%%%%%%%

%%%%%%%%%%%%%%%%%%%%%%%%%%%%%%%%
\subsection{Notation}
Let us first precisely define the $\Gnp$ binomial random graph. $\Gnp$ is a distribution over the class of graphs with the set of nodes $[n]:=\{1,\ldots,n\}$ in which every pair $\{i,j\} \in \binom{[n]}{2}$ appears independently as an edge in $G$ with probability~$p$. Note that $p=p(n)$ may (and usually does) tend to zero as $n$ tends to infinity. We say that $\Gnp$ has some property \emph{asymptotically almost surely} or a.a.s.\ if the probability that $\Gnp$ has this property tends to $1$ as $n$ goes to infinity. For more about this model see, for example,~\cite{Bollobas,JLR,Karonski_Frieze}.

Given two functions $f=f(n)$ and $g=g(n)$, we will write $f(n)=\bigo(g(n))$ if there exists an absolute constant $c \in \R_+$ such that $|f(n)| \leq c|g(n)|$ for all $n$, $f(n)=\Omega(g(n))$ if $g(n)=\bigo(f(n))$, $f(n)=\Theta(g(n))$ if $f(n)=\bigo(g(n))$ and $f(n)=\Omega(g(n))$, and we write $f(n)=o(g(n))$ or $f(n) \ll g(n)$ if $\lim_{n\to\infty} f(n)/g(n)=0$. In addition, we write $f(n) \gg g(n)$ if $g(n)=o(f(n))$ and we write $f(n) \sim g(n)$ if $f(n)=(1+o(1))g(n)$, that is, $\lim_{n\to\infty} f(n)/g(n)=1$.

\subsection{Concentration Tools}\label{sec:concentration}

In this section, we state a few specific instances of Chernoff's bound that we will find useful. Let $(Z_1, \ldots, Z_n)$ be a sequence of independent ${\rm Bernoulli}(p)$ random variables. For each $j \in [n]$, let $X_j = \sum_{i=1}^j Z_i$. In particular, $X_n \in \textrm{Bin}(n,p)$ is a random variable distributed according to a Binomial distribution with parameters $n$ and $p$. Then, a consequence of \emph{Chernoff's bound} (see e.g.~\cite[Theorem~2.1]{JLR}) is that for any $\tau \ge 0$ we have
\begin{eqnarray}
\Prob( X_n - \E [X_n] \ge \tau ) &\le& \exp \left( - \frac {\tau^2}{2 (\E [X_n] + \tau/3)} \right)  \label{chern1} \\
\Prob( \E [X_n] - X_n \ge \tau ) &\le& \exp \left( - \frac {\tau^2}{2 \E [X_n]} \right).\label{chern}
\end{eqnarray}

Moreover, let us mention that the above bounds hold in a more general setting as well, that is, for any sequence $(Z_j)_{1\le j\le n}$ of independent random variables such that for every $j \in [n]$ we have $Z_j \in \textrm{Bernoulli}(p_j)$ with (possibly) different $p_j$-s (again, see~e.g.~\cite{JLR} for more details). 

Finally, we note that $X_n - \E [X_n]$ in~(\ref{chern1}) can be replaced with $\max_{1 \le j \le n} (X_j - \E [X_j])$ and $\E [X_n] - X_n$ in~(\ref{chern}) can be replaced with $\max_{1 \le j \le n} (\E [X_j] - X_j)$. That is, we have
\begin{eqnarray}
\Prob( \max_{1 \le j \le n}(X_j - \E [X_j]) \ge \tau ) &\le& \exp \left( - \frac {\tau^2}{2 (\E [X_n] + \tau/3)} \right)  \label{max-chern1} \\
\Prob( \max_{1 \le j \le n}(\E [X_j] - X_j) \ge \tau ) &\le& \exp \left( - \frac {\tau^2}{2 \E [X_n]} \right).\label{max-chern}
\end{eqnarray}
This is a consequence of a standard martingale bound (see~e.g.~\cite{mcdiarmid1998concentration} for more details). 

%%%%%%%%%%%%%%%%%%%%%%%%%%%%%%%%
\subsection{Coupling}\label{sec:coupling}

Suppose that at some point of the process, the public announcement is captured by $C^t(v)$, $v \in V$. Let $\hat{C}^t(v)$ be any sequence of opinions such that the following properties hold: 
(a) if $\hat{C}^t(v) = 1$, then $C^t(v)=1$,
(b) if $\hat{C}^t(v) = 0$, then $C^t(v) \in \{0, 1,\perp \}$,
(c) if $\hat{C}^t(v) = \perp$, then $C^t(v) = \perp$.
In other words, we get the auxiliary sequence $\hat C^t(v)$ by modifying some of the opinions~1 and $\perp$ in $C^t(v)$ to~0. Hence, the process starting from $C^t(v)$ can be coupled with the auxiliary process starting from $\hat{C}^t(v)$ such that all the properties (a)--(c) are satisfied in every step of the process. In particular, if the auxiliary process converges to all nodes having opinion~1, then so does the original process. This easy observation will turn out to be useful in analyzing the process. 

Similarly, suppose private beliefs in the auxiliary process are dominated by private beliefs in the original process: for any $v \in V$, $\hat{X}(v) \le X(v)$. If the two processes are coupled, then properties (a)--(c) hold again. As before, if the auxiliary process converges to all nodes having opinion~1, then so does the original process. In particular, as mentioned above, the assumption that $\delta \in (0,1/10]$ in Theorems~\ref{thm:not_very_sparse_graphs} and~\ref{thm:very_sparse_graphs} can be relaxed to $\delta \in (0,1/2)$.

%%%%%%%%%%%%%%%%%%%%%%%%%%%%%%%%
\section{Sparse Random Graphs}\label{sec:sparse_graphs}
%%%%%%%%%%%%%%%%%%%%%%%%%%%%%%%%

In this section, we consider sparse random graphs, that is, we will assume that $p=o(1)$. Let $\omega=\omega(n)$ be a function that tends to infinity as $n \to \infty$, arbitrarily slowly. In particular, each time we refer to $\omega$, we will assume that $\omega \ll pn$ and $\omega \ll (1/p)^{1/2}$ so that $1/p \gg 1/ (p\omega) \gg 1/ (p\omega^2) \gg 1$. 

We will consider a few phases. During the first phase (Subsection~\ref{sec:phase1}), most of the nodes that are chosen have not yet announced their opinions ($C^{t-1}(v^t) = \perp$) and none of their neighbours have announced ($N^{t-1}_1(v^t) = N^{t-1}_0(v^t) = 0$). Hence, the announcement of $v^t$ will typically coincide with its private belief. During the second phase (Subsection~\ref{sec:phase2}), it is still the case that most selected nodes are selected for the first time but this time they might have neighbours that announced their opinions. As a result, the argument is more involved but the conclusion is that at the end of the second phase more nodes have correct opinion than not. 

The analysis of the first two phases can be applied for all sparse graphs, even below the threshold for connectivity. The analysis of the final steps of the process is slightly more involved. We first present an easy argument for not very sparse graphs (Subsection~\ref{sec:not_very_sparse_graphs}), that is, when the asymptotic expected degree degree satisfies $pn \gg \log n$. Very sparse graphs for which $pn = \Theta(\log n)$ (but, of course, above the connectivity threshold) are considered in Subsection~\ref{sec:very_sparse_graphs}.

\paragraph{Overview.} A key phenomena in asynchronous dynamics is that the process involves both information diffusion and conventional social learning. Intuitively, the process initially produces some independent beliefs/opinions pop up sporadically throughout the network. These opinions then diffuse in the network during the process as more nodes are selected to announce/update their opinion by learning from their neighbours. With this in mind, our analysis considers multiple phases of the process. We provide a brief description of the different phases below.
\begin{itemize}
    \item \textbf{Phase 1.} In the first few time steps, most nodes that are selected to announce have not been selected earlier and, more importantly, do not have neighbours who have been selected before. So almost all of the opinions in the network at the end of phase one are just the independent private beliefs of the selected nodes. Since the private signals are biased towards being correct, a strict majority of the opinions are correct at the end of the first phase.  In particular, we show that at time $T_1 = \delta/2p$, the number of nodes with opinion $1$ is at least $(1/2 + 3\delta/5)T_1$ and opinion $0$ is represented at most $(1/2 - 3\delta/5)T_1$ times. Moreover, $T_1 (1 - o(1))$ nodes have made some announcement in this phase, that is, very few nodes were selected more than once.
    \item \textbf{Phase 2.} In the second phase, again most nodes that are selected to announce have not been selected earlier. In particular, we show that at any time $t$ during the second phase (i.e., after time $T_1$ but before time $T_2 = n/\omega$), the number of nodes that were selected twice before time $t$ is $o(t)$. Moreover, since a \emph{super majority} of the opinions at the end of the previous phase were correct, we prove that nodes that are selected to announce for the first time are more likely to learn the correct opinion even if we pretend that the few nodes that are selected \emph{again} were to change their opinion to $0$. 
    \item \textbf{Phase 3 (a).} For not very sparse graphs, we are able to show all nodes which were not selected in the first two phases have more neighbours with opinion 1 than not. Again, very few nodes who were selected before are selected again before time $T_3=n/\sqrt \omega$, so even if all of them announce $0$ all nodes who make their first announcement between time $T_2$ and time $T_3$ announce the correct opinion. Finally, even if all the nodes that were selected before time $T_2$ are to have opinion $0$ and all nodes that were selected for the first time between time $T_2$ and time $T_3$ have opinion $1$, we show that a.a.s.\ all announcements after time $T_3$ are always correct.
    \item \textbf{Phase 3 (b).} For very sparse graphs, the proof of the last phase is more involved as there might be nodes whose degree is too small to guarantee that a majority of their neighbours have opinion $1$, even though there is a super majority of opinion $1$ in the network. However, we may bound the number of nodes with small degrees and show that no large degree node has more than one small degree neighbour. With this in hand, we show that after every batch of $O(n\log n)$ many time steps the number of large degree nodes with opinion $0$ shrinks by at least $(\log\log n)^{1/4}$ factor. Hence, after $o(\log n)$ many such batches all large nodes have opinion $1$. Finally, we show that no two small degree nodes are adjacent to each other, and hence all the small degree nodes will also switch to opinion $1$ by copying the opinions of their large degree neighbours.
\end{itemize}

We highlight a few simple techniques that help us in the analysis. Firstly, separating the randomness of the graph, the node selection process and the opinion formation. For example, we wait to reveal/expose the edges adjacent to a node only when she is selected to announce for the first time. This simple, well-known technique, known as the principle of deferred decisions, is very powerful and often used in analysis of randomized algorithms. The idea behind the principle is that the entire set of random choices are not made in advance, but rather fixed only as they are revealed to the algorithm. Second, considering an auxiliary dynamics that is coupled with the actual dynamics in order to ignore problematic but rare events such as the repeated nodes in the first two phases. Finally, finding independent sequences of random variables that stochastically dominate the opinion dynamics sequence in order to compute probability bounds more easily. 

%%%%%%%%%%%%%%%%%%%%%%%%%%%%%%%%
\subsection{Phase 1: $T_1 = \delta/(2p)$}\label{sec:phase1}

In the analysis of the process, it will be convenient to ignore opinions of a small fraction of nodes, and consider the following \emph{auxiliary dynamics}. We will use $D^t(v) \in \{\perp, ?, 0, 1\}$ to denote the \emph{auxiliary announcement} of $v \in V$ at time $t$. For any $i \in \{\perp, ?, 0, 1\}$, let $Z_i^t$ be the number of nodes that have auxiliary opinion $i$ at time $t$, that is, $Z_i^t = | \{ v \in V : D^t(v) = i \} |$. We will explain how the values of $D^t(v)$ are determined soon but the auxiliary dynamics will be coupled with the original one and, in particular, we will make sure that the following property holds.

\begin{property}\label{prop:D}
If $D^t(v) = i$ for some $i \in \{\perp, 0, 1\}$ and time $t$, then $C^t(v)=D^t(v)$. On the other hand, if $D^t(v) = ?$, then $C^t(v) \in \{0, 1\}$. As a result, for $i \in \{0, 1\}$ and any time $t$ during the first phase, we have
\begin{equation}\label{eq:relation_YZ}
Z^t_i \le Y^t_i \le Z^t_i + Z^t_?.
\end{equation}
\end{property}

The first phase takes $T_1 = \delta/(2p) = \Theta(1/p) \gg \omega^2 \gg 1$ rounds. In order to keep the analysis easy, we postpone exposing edges of $\Gnp$ for as long as possible, and keep the following useful property.

\begin{property}\label{prop:exposingGnp}
At any time $t$, only edges of $\Gnp$ with both endpoints in the set $\{ v : D^t(v) \neq \perp \}$ are exposed. 
\end{property}

The auxiliary dynamics, coupled with the original one, that we aim to understand is defined as follows. Consider a node $v^t$ chosen at time $t$. For all other nodes $v\neq v^t$ we have $D^{t}(v) = D^{t-1}(v)$. For $v^t$ we have,
$$
D^t(v^t) = 
\begin{cases}
? & \text{ if } D^{t-1}(v^t) \neq \perp,\\
? & \text{ if } \exists\ { \text{ node } v \text{ such that } v \in N(v^t) \text{ and } } D^{t-1}(v) \neq \perp,\\
X(v^t) & \text{ otherwise}.
\end{cases}
$$
That is, if $v^t$ had announced her opinion at least once before time $t$ ($D^{t-1}(v^t), C^{t-1}(v^t) \neq \perp$), then we fix $D^{t}(v^t) = ?$. On the other had, if $v^t$ has not announced her opinion yet (that is, $D^{t-1}(v^t) = C^{t-1}(v^t) = \perp$), then we expose edges of $\Gnp$ between $v^t$ and the set $\{ v : D^{t-1}(v) \neq \perp \}$. If no edge between $v^t$ and the set $\{ v : D^{t-1}(v) \neq \perp \}$ is present, then no neighbour of $v^t$ has an announced opinion and so $D^t(v^t) = C^t(v^t) = X(v^t)$ is fixed to the private belief of $v^t$. Otherwise (that is, at least one edge is present), then we simply fix $D^t(v^t) = ?$. Let us note that, {an alternative approach would be} to investigate the value of $C^t(v^t)$ and then fix $D^t(v^t) = C^t(v^t)$. However, we expect at most $pt \le pT_1 = \delta/2$ edges between $v^t$ and $\{ v : D^{t-1}(v) \neq \perp \}$, and so there will not be many nodes $v^t$ of this type. As a result, we may simply ignore the announcements of such nodes, thus simplifying our analysis. 

Moreover, a useful implication of this approach is that in order to estimate the values of $Z^t_{\perp}$ and $Z^t_?$ in this process, we do not need to uncover nodes' private believes ($X(v)$'s). Hence, we may postpone exposing private beliefs of nodes with $D^t(v) \not\in \{ \perp, ? \}$ to the very end of this phase, and only then expose this information to determine how many nodes satisfy $D^{T_1}(v) = 1$ and how many of them satisfy $D^{T_1}(v) = 0$. Finally, it is easy to see that  Property~\ref{prop:D} is satisfied at time $T_1$ and Property~\ref{prop:exposingGnp} is satisfied in any point of the process.

Here is the main result of this subsection. 

\begin{proposition}\label{prop:end_of_phase1}
Suppose that $p = p(n) \ll 1$ and $p \gg 1/n$. Set $T_1 = \delta/(2p)$. Let $\omega=\omega(n) \ll \min\{ pn, (1/p)^{1/2}\}$ be any function that tends to infinity as $n \to \infty$. Then, a.a.s.\ the following holds:
\begin{eqnarray}
Z_?^{T_1} &\le& \frac {\delta T_1}{4} \left( 1 + \bigo( 1/\omega) \right) \label{eq:phase1_z?} \\
Z_1^{T_1} &\ge& (1/2+3\delta/5) \ T_1 \label{eq:phase1_z1} \\
Z_?^{T_1} + Z_1^{T_1} + Z_0^{T_1} &=& T_1 \left( 1 - \bigo( 1/\omega) \right). \label{eq:phase1_sumofz}
\end{eqnarray}
As a result, by Property~\ref{prop:D},
\begin{eqnarray*}
Y_1^{T_1} &\ge& (1/2+3\delta/5) \ T_1 \\
Y_0^{T_1} &\le& (1/2-3\delta/5) \ T_1 \\
Y_1^{T_1} + Y_0^{T_1} &=& T_1 \left( 1 - \bigo( 1/\omega) \right).
\end{eqnarray*}
\end{proposition}

\begin{proof}
Let us start with investigating $Z_?^{T_1}$. Recall that in our auxiliary dynamics, there are two ways node $v^t$ could change its state to $D^{t}(v^t) = ?$ at time $t$. Let $I_t$ be the indicator random variable that this happens because $D^{t-1}(v^t) \neq \perp$, and let $I = \sum_{t=1}^{T_1} I_t$. Similarly, let $J_t$ be the indicator random variable that $D^{t-1}(v^t) = \perp$ but there is an edge between $v^t$ and the set $\{ v : D^{t-1}(v) \neq \perp \}$. Let $J = \sum_{t=1}^{T_1} J_t$.

Note that, at most $t-1$ distinct nodes have made an announcement before round $t$. In particular, at most one node can change its state from $D^{t-1}(v) = \perp$ to  $D^{t}(v) \neq \perp$,  deterministically, at any round $t$ of the process. 
So, the number of nodes with $D^{t-1}(v) \neq \perp$ is $n-Z^{t-1}_{\perp} \le t-1$. We get that
$$
\Pr ( I_t = 1 ) = \frac {n-Z^{t-1}_{\perp}} {n} \le \frac {t-1}{n},
$$
and so $I$ can be stochastically upper bound by $\hat{I}=\sum_{t=1}^{T_1} \hat{I}_t$ where $(\hat{I}_t)_{1\le t \le T_1}$ are independent variables and for every $t \in [T_1]$ we have $\hat{I}_t \in \textrm{Bernoulli}((t-1)/n)$. Note that, since $pn \gg \omega$, 
\begin{equation}\label{eq:nodes-selected-again}
    \E [ \hat{I} ] = \sum_{t=1}^{T_1} \frac {t-1}{n} = \frac {(T_1-1)T_1}{2n} \sim \frac {\delta T_1}{4 p n} \ll \frac {T_1}{\omega}.
\end{equation}
It follows from Chernoff's bound~(\refeq{chern1}) (and the comment right after it) applied with $t = T_1 / \omega = \Theta(1/(p\omega)) \gg \omega \gg 1$ that
$$
\Pr ( \hat{I} \ge \E[ \hat{I} ] + t ) \le \exp \left( - \frac {t^2}{(2/3+o(1)) t} \right) = \exp \left( - \Theta(t) \right) = o(1).
$$
So a.a.s.\ $I \le \hat{I} = \bigo(T_1 / \omega)$. Similarly, since $pt \le pT_1 = \delta / 2 < 1/4$,
$$
\Pr ( J_t = 1 ) = \frac {Z^{t-1}_{\perp}} {n} \left( 1 - (1-p)^{n-Z^{t-1}_{\perp}} \right) \le 1 - (1-p)^t  = 1 - \left( 1 - pt + p^2 \binom{t}{2} - \ldots \right) \le pt.
$$
As before, we stochastically upper bound $J$ by $\hat{J}=\sum_{t=1}^{T_1} \hat{J}_t$, where $\hat{J}_t \in \textrm{Bernoulli}(pt)$. We get that
$$
\E [ \hat{J} ] = \sum_{t=1}^{T_1} pt = \frac {p(T_1+1)T_1}{2} = \frac {pT_1^2}{2} \left( 1 + \bigo(1/T_1) \right) =  \frac {\delta T_1}{4} \left( 1 + \bigo(1/\omega) \right),
$$
and Chernoff's bound~(\refeq{chern1}) (applied with $t=\E [ \hat{J} ] / \omega$) implies that 
$$
\Pr ( \hat{J} \ge \E[ \hat{J} ] + t ) \le \exp \left( - \frac {\E[ \hat{J} ]}{ (2+o(1)) \omega^2} \right) = \exp \left( - \Theta(T_1/\omega^2) \right) = \exp \left( - \Theta(1/(p\omega^2)) \right) = o(1).
$$
Hence, a.a.s.\ $J \le \hat{J} \le \frac {\delta T_1}{4} \left( 1 + \bigo( 1/\omega) \right)$ and so a.a.s.\ $Z_?^{T_1} \le I + J \le \frac {\delta T_1}{4} \left( 1 + \bigo( 1/\omega) \right)$. This proves~(\ref{eq:phase1_z?}).

It remains to investigate $Z_0^{T_1}$ and $Z_1^{T_1}$. Let us summarize the situation at time $T_1$. The number of rounds when nodes were not chosen for the first time is at most $I = \bigo(T_1 / \omega)$ a.a.s. Hence, a.a.s.\ the number of nodes that were chosen at least once is equal to $T_1 - \bigo(T_1 / \omega)$. This proves~(\ref{eq:phase1_sumofz}). Moreover, it implies that a.a.s.\ the number of nodes with $D^{T_1}(v) \not\in \{ \perp, ? \}$ is equal to
$$
Z_1^{T_1} + Z_0^{T_1} = T_1 - \bigo(T_1 / \omega) - Z_?^{T_1} \ge (1-\delta/4) T_1 \left( 1 + \bigo( 1/\omega) \right).
$$
More importantly, as mentioned above, in the analysis so far we did not use their opinions which are consistent with their private beliefs. We conveniently deferred this information up to now. After exposing this information, we get that $Z_1^{T_1}$ is stochastically lower bounded by the random variable $\hat{Z}_1 \in \textrm{Bin}((1-\delta/4) T_1 - c T_1 / \omega, 1/2+\delta)$, where $c>0$ is a large enough constant. After applying Chernoff's bound~(\refeq{chern}) (with $t = T_1/\omega$) we get that 
\begin{eqnarray*}
Z_1^{T_1} \ge \hat{Z}_1 &=& (1/2+\delta) (1-\delta/4) T_1 (1+\bigo(1/\omega)) \\
&\ge& (1/2+\delta - \delta/4) T_1 (1+\bigo(1/\omega)) \\
&\ge& (1/2 + 3\delta/5) T_1
\end{eqnarray*}
with probability at least 
$$
1 - \exp( - \Theta( T_1 / \omega^2 )) = 1 - \exp( - \Theta( 1 / (p\omega^2) )) = 1 - o(1).
$$
This proves~(\ref{eq:phase1_z1}). 

The conclusion for $Y_1^{T_1}$ follows immediately from Property~\ref{prop:D}, and the bound for $Y_0^{T_1}$ is a trivial implication of the fact that $Y_1^{T_1} + Y_0^{T_1} \le T_1$. The proof of the proposition is finished.
\end{proof}

%%%%%%%%%%%%%%%%%%%%%%%%%%%%%%%%
\subsection{Phase 2: $T_2 = T_2(n)$ such that $\omega/p \le T_2 \le n/\omega$}\label{sec:phase2}

By Proposition~\ref{prop:end_of_phase1}, since we aim for a statement that holds a.a.s., we may assume that at the beginning of Phase~2, 
\begin{eqnarray}
Y_1^{T_1} &\ge& (1/2+3\delta/5) \ T_1 \nonumber \\
Y_0^{T_1} &\le& (1/2-3\delta/5) \ T_1 \nonumber \\
Y_1^{T_1} + Y_0^{T_1} &=& T_1 \left( 1 + \bigo( 1/\omega) \right). \label{eq:beginning_of_phase2}
\end{eqnarray}
As in the previous phase, it will be convenient to ignore opinions of some problematic nodes and assign auxiliary announcements $D^t(v)=?$ to such nodes. We will continue using $Z_i^t$ to denote the number of nodes that have auxiliary opinion $i$ at time $t$. We fix $D^{T_1}(v) = C^{T_1}(v)$ for all $v$ so, initially, auxiliary announcements coincide with the truth announcements. However, this time we assign $D^t(v^t)=?$ only if $D^{t-1}(v^t) \neq \perp$ (that is, the node chosen at time $t$ has made an announcement in the past); otherwise, the auxiliary announcement $D^t(v^t)$ is determined immediately pretending that all neighbours $v$ of $v^t$ with $D^{t-1}(v) = ?$ announced 0. More formally, for each node $v$ and $i\in \{0,1,\perp,?\}$ let $\hat N_i^{t}(v)$ denote the number of neighbours $v'$ of $v$ with auxiliary opinion $D^{t}(v')= i$ at time $t$. Then we have,
$$
D^t(v^t) = 
\begin{cases}
? & \text{ if } D^{t-1}(v^t) \neq \perp,\\
1 & \text{ if } \hat N_1^{t-1}(v^t) > \hat N_0^{t-1}(v^t) + \hat N_?^{t-1}(v^t),\\
0 & \text{ if } \hat N_1^{t-1}(v^t) < \hat N_0^{t-1}(v^t) + \hat N_?^{t-1}(v^t),\\
X(v^t) & \text{ if } \hat N_1^{t-1}(v^t) = \hat N_0^{t-1}(v^t) + \hat N_?^{t-1}(v^t).
\end{cases}
$$
As a consequence, $D^t(v)$ and $C^t(v)$ are coupled so that the following property is satisfied.

\begin{property}\label{prop:D'}
If $D^t(v) = i$ for some $i \in \{\perp, 1\}$ and time $t$, then $C^t(v)=D^t(v)$. On the other hand, if $D^t(v) = i$ for some $i \in \{0, ?\}$, then $C^t(v) \in \{0, 1\}$. As a result, for any time $t$ during the second phase, we have $Y^t_1 \ge Z^t_1$.
\end{property}

As before, it is easy to see that Property~\ref{prop:exposingGnp} is also satisfied during this phase. Here is the main result of this subsection. 

\begin{proposition}\label{prop:end_of_phase2}
Suppose that $p = p(n) \ll 1$ and $p \gg 1/n$. Let $\omega=\omega(n) \ll \min\{ pn, (1/p)^{1/2}\}$ be any function that tends to infinity as $n \to \infty$. Set $T_2 = T_2(n)$ such that $\omega/p \le T_2 \le n/\omega$. Then, a.a.s.\ the following holds:
\begin{eqnarray}
Z_?^{T_2} &=& \bigo( T_2/\omega)  \nonumber \label{eq:phase2_z?} \\
Z_1^{T_2} &\ge& (1/2+\delta/2) \ T_2 \nonumber \label{eq:phase2_z1} \\
Z_?^{T_2} + Z_1^{T_2} + Z_0^{T_2} &=& T_2 \left( 1 - \bigo( 1/\omega) \right).  \nonumber \label{eq:phase2_sumofz}
\end{eqnarray}
As a result, by Property~\ref{prop:D'},
\begin{eqnarray*}
Y_1^{T_2} &\ge& (1/2+\delta/2) \ T_2 \\
Y_0^{T_2} &\le& (1/2-\delta/2) \ T_2 \\
Y_1^{T_2} + Y_0^{T_2} &=& T_2 \left( 1 - \bigo( 1/\omega) \right).
\end{eqnarray*}
\end{proposition}

Before we move to the proof of this proposition, let us make some simple but useful observations. First, note that only a negligible fraction of the nodes have an opinion that we do not control.

\begin{lemma}\label{lem:repetitions_phase2}
Suppose that $p = p(n) \ll 1$ and $p \gg 1/n$. Let $\omega=\omega(n) \ll \min\{ pn, (1/p)^{1/2}\}$ be any function that tends to infinity as $n \to \infty$. Set $T_2 = T_2(n)$ such that $\omega/p \le T_2 \le n/\omega$. Then, a.a.s., for any $t$ such that $T_1 \le t \le T_2$, $Z_?^t \le 2t/\omega$.
\end{lemma}
\begin{proof}
In fact, we will prove something stronger. Let $X_t$ be the number of nodes that were selected at least two times up to time $t$ (which could happen before time $T_1$). We will prove that a.a.s.\ for any $1 \le t \le n/\omega$, $X_t \le 2t/\omega$. 

Case 1: $1 \le t \le n^{2/5}$. As argued in the proof of Proposition~\ref{prop:end_of_phase1} (see~\eqref{eq:nodes-selected-again}), one can bound the expected value of $X_{n^{2/5}}$ as follows:
$$
\E[ X_{n^{2/5}} ] \le \sum_{t \le n^{2/5}} \frac {t-1}{n} \sim \frac {n^{4/5}}{2n} = o(1).
$$
Since $X_t$ is non-decreasing, it follows immediately from Markov's inequality that a.a.s.\ $X_t \le X_{n^{2/5}} \le \E[X_{n^{2/5}}] n^{1/5} < 1$ for all $t$ such that $1 \le t \le n^{2/5}$.

Case 2: $n^{2/5} \le t \le n^{3/5}$. As before, we observe that $\E[ X_{n^{3/5}} ] = \bigo(n^{6/5}/n) = \bigo( n^{1/5} )$ and so using Markov's inequality again we get that a.a.s.\ $X_t \le X_{n^{3/5}} \le \E [X_{n^{3/5}}] \log n = \bigo( n^{1/5} \log n) \le 2t / \omega$ for all $t$ such that $n^{1/3} \le t \le n^{2/3}$.

Case 3: $n^{3/5} \le t \le n/\omega$. Fix any $t$ in this range. This time we stochastically upper bound $X_t$ by $X'_t \in \textrm{Bin}(t,t/n)$ with $\E [X'_t] = t^2/n \ge n^{1/5} \gg \log n$. Chernoff's bound~(\ref{chern1}) implies that $X_t \le X'_t \le 2 \E[X'_t] = 2t^2/n \le 2t/\omega$ with probability $1-\bigo(n^{-2})$. The desired result holds by the union bound over all $t$ in this range. 
\end{proof}

Let us fix $k \in \N$ and consider random variable $X_k \in \textrm{Bin}(k,1/2+\delta/2)$. We will need to understand the following sequence of constants (the connection to our problem will become clear soon): 
\begin{equation}\label{def:q}
q_k := \Prob( X_k > k/2 ) + \Prob ( X_k = k/2 ) \cdot (1/2+\delta).
\end{equation}
Clearly, $q_0 = 1/2 + \delta$ and $q_1 = 1/2 + \delta/2$. For any other value of $k \ge 2$, $q_k \ge 1/2 + 51\delta/100$ as we show in the next technical lemma. The proof of this fact can be found in the appendix. 
\begin{lemma}\label{lem:binomial}
Fix $k \in \N$ such that $k \ge 2$, and $\delta \in (0,1/10]$. Then, 
$$
q_k \ge \frac 12 + \frac {51}{100} \delta.
$$
\end{lemma}

Now, we are ready to go back to analyzing the behaviour of the process during the second phase. 

\begin{proof}[Proof of Proposition~\ref{prop:end_of_phase2}]
Our goal is to show that a.a.s.\ the following inequalities hold for any $t$ such that $T_1 \le t \le T_2$:
\begin{eqnarray}
\frac {Z_1^t}{Z_0^t + Z_?^t} &\ge& \frac {1/2+\delta/2}{1/2-\delta/2} \label{eq:ratio_phase2} \\
Z_?^t &\le& 2t / \omega. \label{eq:ratio_phase2'}
\end{eqnarray}
Formally, we define the stopping time $S$ to be the minimum value of $t \ge T_1$ such that either~(\ref{eq:ratio_phase2}) fails,~(\ref{eq:ratio_phase2'}) fails or $t = T_2$. (A stopping time is any random variable $S$ with values in $\{T_1, T_1+1, \ldots, T_2\}$ such that, for any time $\hat t$, it is determined whether $S = \hat{t}$ from knowledge of the process up to and including time $\hat{t}$.)

Property~(\ref{eq:ratio_phase2'}) is trivially satisfied at the beginning of the second phase as $Z_?^{T_1} = 0$. (Recall that we used Property~\ref{prop:D} to replace all auxiliary opinions $?$s by $0$s so at the beginning of the second phase there are only opinions $1$ and $0$ present---see~(\ref{eq:beginning_of_phase2}).) By Proposition~\ref{prop:end_of_phase1}, since we aim for a statement that holds a.a.s., we may assume that~(\ref{eq:ratio_phase2}) is satisfied at the beginning of the second phase. In fact, 
$$
\frac {Z_1^{T_1}}{Z_0^{T_1} + Z_?^{T_1}} {= \frac {Y_1^{T_1}}{Y_0^{T_1} + 0}} \ge \frac {1/2+101\delta/200}{1/2-101\delta/200} \ge \frac {1/2+\delta/2}{1/2-\delta/2}.
$$
It will be convenient to define $Z^t = Z_1^t + Z_0^t + Z_?^t$; that is, $Z^t$ is the number of nodes that announced their opinions {by time $t$}. If~(\ref{eq:ratio_phase2'}) is satisfied, then only a negligible fraction of nodes were selected more than once and we get that $Z^t = t (1- \bigo(1/\omega)) \sim t$.

Let us first show that if~(\ref{eq:ratio_phase2}) and~(\ref{eq:ratio_phase2'}) are satisfied at time $t$ and the node selected at {time} $t+1$ was not selected before (that is, $D^{t}(v^{t+1})=C^{t}(v^{t+1}) = \perp$), then the probability that $v^{t+1}$ announces an auxiliary opinion~1 is at least $1/2 + 101/200\delta$.

We first expose edges from $v^{t+1}$ to the set $\{ v : D^{t}(v) \neq \perp \}$ (see Property~\ref{prop:exposingGnp}) and let us define $p_k$ to be the probability that $v^{t+1}$ has precisely $k$ neighbours in that set. In particular, we have 
\begin{eqnarray}
    p_1  & = Z^tp(1-p)^{Z^t - 1} & = \lambda(1-p)^{\lambda/p -1} \notag\\
    && \le \lambda e^{-\lambda}/(1-p)\notag\\
    && \le 1/e + o(1) ~<~ 1/2, \label{eq:bound_p1}
\end{eqnarray}
where $\lambda = pZ^t = pt(1-\bigo(1/\omega))$ and the second inequality follows because $xe^{-x} \le e^{-1}$ and $p = o(1)$. 

Now, condition on $v^t$ having exactly $k$ neighbours that already announced their opinion. Note that we did not expose the neighbours yet (only the number of them) so neighbours form a random set of cardinality $k$ from the set $\{ v : D^{t-1}(v) \neq \perp \}$. Let $r_k$ to be the probability that $v^t$ announces auxiliary opinion 1 in this conditional probability space. It happens if more than $k/2$ neighbours of $v^t$ have $D^{t-1}(v)=1$. Moreover, if exactly $k/2$ neighbours have this property, then $v^t$ announces opinion~1 with probability $1/2+\delta$, which is the probability that its private belief is~1. Since~(\ref{eq:ratio_phase2}) holds, $r_k$ can be lower bounded by $q_k$ which we defined in~(\ref{def:q}). It follows that the probability that $v^t$ announces~1 is asymptotic to 
\begin{eqnarray}
\sum_{k \ge 0} r_k \cdot p_k &\ge& \sum_{k \ge 0} q_k \cdot p_k = q_1 p_1 +  \sum_{k \ge 0, k\neq 1} q_k \cdot p_k \nonumber\\
&\ge& \left( \frac 12 + \frac {\delta}{2} \right) p_1 + \left( \frac 12 + \frac {51}{100} \delta \right) (1-p_1) \nonumber \\
&=& \left( \frac 12 + \frac {51}{100}\delta \right)  - p_1 \left(\frac {1}{100} \delta \right) \nonumber\\
&\ge& \frac 12 + \frac {101}{200} \delta, \label{eq:lower_bound_p}
\end{eqnarray}
where the second inequality follows from Lemma~\ref{lem:binomial} and $q_0 = (1/2 +\delta) \ge (1/2 + 51\delta/100)$, and the last one from~(\ref{eq:bound_p1}).

Let $s$ be the number of rounds $t$ in the second phase {in which} $v^t$ was not selected before, i.e., $D^{t-1}(v^t) = \perp$, and let $t_1,t_2,\ldots, t_s$ denote such rounds.
Clearly, $s \le T_2 - T_1 = T_2 (1 - \bigo(1/\omega))$ but, in fact, a.a.s.\ we have $s = T_2 (1 - \bigo(1/\omega))$ by Lemma~\ref{lem:repetitions_phase2}. Indeed, it follows from Lemma~\ref{lem:repetitions_phase2} that a.a.s.\ the number of rounds in which $v^t$ was selected before is at most $2 T_2 / \omega$.
For $i \in [s]$, let $L_i$ be the indicator random variable for the event that $v^{t_i}$ announced an auxiliary opinion~1, that is, $L_i = Z_1^{t_i}-Z_1^{t_i-1}$.
If both~(\ref{eq:ratio_phase2}) and~(\ref{eq:ratio_phase2'}) hold at time $t_i-1$, then $\Prob(L_i = 1) \ge 1/2 + 101 \delta / 200$ but, of course, we cannot condition on these two properties to hold. {Instead, we will use a small trick and consider an auxiliary sequence of random variables after the stopping time $S$ when one of the properties fails.} 

Fix $\hat p=1/2 + 101 \delta / 200$ and let $M_1, \ldots, M_s$ be a sequence of independent Bernoulli variables with parameter $\hat p$. For each $i \in [s]$, we define $L'_i = L_i$ if both~(\ref{eq:ratio_phase2}) and~(\ref{eq:ratio_phase2'}) hold at times $t < t_i$ and otherwise $L'_i = M_i$. That is, the process ``stops'' at our stopping time $S$ which, in our context, means that it simply follows {part of the sequence $(M_i)_{i=1}^s$ (namely, $(M_i)_{i=S+1}^s$)} from that point on, ignoring the behaviour of the original process. Thus, defining $L'_{\le j} = \sum_{i=1}^{j} L'_i$ and $M_{\le j} = \sum_{i=1}^{j} M_i$, (\ref{eq:lower_bound_p}) implies that one can couple $L'_{\le j}$ and $M_{\le j}$ such that $L'_{\le j} \ge M_{\le j}$ for all $j \in [s]$.

Note that $\E [M_{\le j}] = \hat p j = (1/2 + 101 \delta / 200)j$ for any $j \in [s]$. If follows from Chernoff's bound~(\ref{max-chern}), 
\begin{align*}
\Prob & \left( \exists_{1 \le j \le s} \ \E [M_{\le j}] - M_{\le j} \ge \frac {\delta}{400} (T_1+j) \right) \\
&\le \sum_{a \ge 1} \ \Prob \left( \max_{(2^{a-1}-1)T_1 < j \le (2^a-1)T_1} \Big( \E [M_{\le j}] - M_{\le j} \Big) \ge 2^{a-1} \frac {\delta}{400} T_1 \right) \\
&\le \sum_{a \ge 1} \ \exp \left( - \Theta( 2^a T_1) \right) = \exp \left( - \Theta( T_1) \right) = o(1),
\end{align*}
since $T_1 = \Theta(1/p) \to \infty$. In other words, a.a.s.\ for any $j \in [s]$,
$$
L'_{\le j} \ge M_{\le j} \ge \left( \frac 12 + \frac {101}{200} \delta \right) j - \frac {\delta}{400} (T_1+j).
$$
Since $L_{\le j} = L'_{\le j}$ for any $j \in [s]$ such that $t_j < S$, and $Z_1^t$ can decrease by at most one in a single round, a.a.s.
\begin{eqnarray*}
Z_1^S &\ge& Z_1^{S-1} - 1 \\
&\ge& Z_1^{T_1} + L_{\le S-T_1-\bigo(S/\omega)} - \bigo(S/\omega) \\
&\ge& \left( \frac 12 + \frac {101}{200} \delta \right) T_1 + \left( \frac 12 + \frac {101}{200} \delta \right) (S-T_1-\bigo(S/\omega)) - \frac {\delta}{400} (S-\bigo(S/\omega)) - \bigo(S/\omega) \\
&\ge& \left( \frac 12 + \frac {201}{400} \delta \right) S - \bigo(S/\omega) \\
&\ge& \left( \frac 12 + \frac {\delta}{2} \right) S,
\end{eqnarray*}
implying that ~\eqref{eq:ratio_phase2} holds at time $S$. Indeed, there were $Z_1^{T_1}$ nodes with auxiliary opinion 1 at the beginning of the second phase. By Lemma~\ref{lem:repetitions_phase2}, $S-T_1-\bigo(S/\omega)$ nodes were selected for the first time before the stopping time and $L_{\le S-T_1-\bigo(S/\omega)}$ of them announced~1 at that time. Finally, at most $\bigo(S/\omega)$ nodes that already announced their opinion were selected again. It implies that a.a.s.\ the process does not ``stop'' because of~(\ref{eq:ratio_phase2}) failing. By Lemma~\ref{lem:repetitions_phase2}, a.a.s.\ it also does not stop because of~(\ref{eq:ratio_phase2'}). Hence, a.a.s.\ $S = T_2$ and the proof of the proposition is finished.
\end{proof}

%%%%%%%%%%%%%%%%%%%%%%%%%%%%%%%%
\subsection{Not Very Sparse Random Graphs}\label{sec:not_very_sparse_graphs}

In this subsection, we provide a relatively easy argument that works for random graphs with $pn \gg \log n$. In particular, we show that a.a.s.\ after round $T_2$ but before round $T_3 = n/\sqrt \omega$ all nodes that are selected for the first time announce $1$. Moreover, after round $T_3$ every node selected announces $1$ a.a.s.

\begin{proof}[Proof of Theorem~\ref{thm:not_very_sparse_graphs}]
Let $\omega=\omega(n) \ll \min\{ (pn/\log n)^{1/2}, pn, (1/p)^{1/2}\}$ be any function that tends to infinity as $n \to \infty$. In particular, $pn \ge \omega^2 \log n$. Fix $T_2 = T_2(n) = n / \omega$. It follows from Proposition~\ref{prop:end_of_phase2} that a.a.s.\ at the end of the second phase, there are $Y_1^{T_2} \ge (1/2+\delta/2)T_2$ nodes that announced opinion~1, and so $Y_0^{T_2} \le (1/2-\delta/2)T_2$ nodes announced opinion~0; moreover, $ Y_1^{T_2}+Y_0^{T_2} = T_2 (1+\bigo(1/\omega))$.

Let $V_i = \{ v : C^t(v) = i \}$ be the set of nodes with opinion~$i \in \{0,1\}$ at time $T_2$. Note that, by Property~\ref{prop:exposingGnp}, we may assume that only edges within $V_0 \cup V_1$ are exposed at that stage of the process. We will first show that a.a.s.\ all nodes $v \notin V_0 \cup V_1$ have substantially more neighbours in $V_1$ than in $V_0$. Indeed, this is a simple consequence of the Chernoff bounds~(\ref{chern1}) and~(\ref{chern}): for any $v \notin V_0 \cup V_1$:
\begin{align*}
\Prob \Big( |N(v) \cap V_1| & \le |N(v) \cap V_0| + \delta T_2 p / 2 \Big) \\
&\leq \Pr \Big( |N(v)\cap V_1| \le (1/2 - \delta/4)T_2p  + \delta T_2p/2 \text{ or } |N(v)\cap V_0| \ge (1/2 - \delta/4)T_2p \Big) \\
& \le \Prob \Big( |N(v) \cap V_1| \le (1/2+\delta/4)T_2p \Big) + \Prob \Big( |N(v) \cap V_0| \ge (1/2-\delta/4)T_2p \Big) \\
&= \Prob \Big( \textrm{Bin}(|V_1|,p) \le (1/2+\delta/4)T_2p \Big) + \Prob \Big( \textrm{Bin}(|V_0|,p) \ge (1/2-\delta/4)T_2p \Big) \\
&\le 2 \exp \Big( - \Theta(T_2p) \Big) \\
&= 2 \exp \left( - \Omega \left( \frac {n}{\omega} \cdot \frac {\omega^2 \log n}{n}  \right) \right) \\
&= \bigo( 1/n^2 ),
\end{align*}
{where the first inequality follows simply by observing that $|N(v)\cap V_1| > c  + \delta T_2p/2$ and $ |N(v)\cap V_0| < c$ implies $|N(v) \cap V_1|  > |N(v) \cap V_0| + \delta T_2 p / 2 $}. The final inequality follows since $\E [ \textrm{Bin}(|V_1|,p) ] \ge (1/2+\delta/2)T_2p$ and $\E [ \textrm{Bin}(|V_0|,p) ] \le (1/2-\delta/2)T_2p$. The desired property holds by the union bound over all nodes $v \notin V_0 \cup V_1$.

Fix $T_3 = T_3(n) = n / \sqrt{\omega}$. The third phase will last till time $T_3$. Let $V_1' \subseteq V_1$ be the set of nodes from $V_1$ that were selected during the third phase. Note that each node from $V_1$ is selected during the third phase with probability at most $(T_3 - T_2)/n \le 1/\sqrt{\omega}$. Hence, $\E [|V_1'|] \le |V_1| / \sqrt{\omega}$ and so a.a.s.\ $|V_1'| \le |V_1| / \omega^{1/3}$ by Markov's inequality. A simple but important observation is that $V_1'$ is determined exclusively by the selection process (coupon collector process); in particular, it does not depend on the random graph nor the opinion dynamics. Hence, we can use Chefnoff's bound again to show that a.a.s.\ all nodes $v \notin V_0 \cup V_1$ have very few neighbours in $V_1'$. Indeed, note that for any $v \notin V_0 \cup V_1$, the number of neighbours of $n$ in $V_1'$ can be stochastically upper bounded by $\textrm{Bin}(|V_1| / \omega^{1/3},p)$ with expectation $|V_1| p / \omega^{1/3} = \Theta(np/\omega^{4/3}) = \Omega(n^{2/3} \log n) \gg \log n$. Hence, $|N(v) \cap V_1'| = \bigo( |V_1| p / \omega^{1/3} ) = \bigo( T_2 p / \omega^{1/3} ) = o(T_2p)$ with probability $1-\bigo(1/n^2)$, and so a.a.s.\ all nodes $v \notin V_0 \cup V_1$ satisfy this property. 

Combining the two properties together, we get that a.a.s.\ for all nodes $v \notin V_0 \cup V_1$ we have 
\begin{equation}\label{eq:property_phase3}
|N(v) \cap (V_1 \setminus V_1')| > |N(v) \cap (V_0 \cup V_1')|. 
\end{equation}
Let $W_1$ be the set of nodes outside of $V_0 \cup V_1$ that were selected during the third phase (possibly multiple times). If property~(\ref{eq:property_phase3}) is satisfied, then (deterministically) all  nodes in $W_1$ announce~1 in this phase. Indeed, even if all nodes from $V_1'$ changed their opinion to~0 in the meantime, nodes in $V_1$ still have majority of their neighbours with opinion~1. 

Let us summarize the situation at the beginning of the fourth (and the last) phase. Recall that $W_1$ consists of nodes that were selected for the first time during the third phase. Let $W_0 = V_0 \cup V_1$ be the set of nodes that were selected before the third phase (that is, during the first or the second phase). A.a.s.\ nodes in $W_1$ have opinion~1 and $|W_1| = (T_3-T_2) + \bigo(T_3^2/n) \sim T_3$. We may assume that nodes in $W_0$ have opinion~0 and a.a.s.\ $|W_0| = T_2 (1+\bigo(1/\omega)) \sim T_2 = o(T_3)$. Again, it is important to notice that $W_1$ and $W_0$ are determined exclusively by the selection process. ($V_1$ and $V_0$ do not posses this property and that was the main reason we needed to consider the third phase.) We may then use Chernoff's bound again, on the number of neighbours in $W_1$ and $W_0$ of any given node, to show that a.a.s.\ all nodes (not only outside of $W_1 \cup W_0$!) have more neighbours in $W_1$ than in $W_0$. It means that every node that is selected during this last phase announces opinion~1. 

Since a.a.s.\ every node is selected at least once during the next $n( \log n + \omega' / 2)$ rounds, the process is over after at most that many rounds with everyone converging to opinion~1. Hence, a.a.s.\ the entire process takes at most $T_3 + n( \log n + \omega' / 2) \le n( \log n + \omega')$ rounds. In fact, the expected number of nodes that were selected before the last phase but were not selected in the first $T'_4 = n ( \log n - \log \omega / 4)$ rounds of the last phase is equal to
$$
T_3 \left( 1 - \frac 1n \right)^{T'_4} \le \frac {n}{\sqrt{\omega}} \exp( - \log n + \frac{1}{4}\log \omega) = \omega^{-1/4} = o(1),
$$
and so a.a.s.\ all nodes selected before the last phase are selected again during the first $T_4'$ rounds of the last phase. On the other hand, a.a.s.\ there are still some nodes not selected at all after $T_3 + T'_4$ rounds. {Indeed, this follows immediately from the well studied coupon collector concentration bound for $\hat{T}$: $\Prob( \hat{T} < n\log n -cn) < e^{-c}$.} The conclusion is that a.a.s.\ all nodes are selected at least once between round $T_3$ and $\hat{T}$, and the proof is finished.
\end{proof}

%%%%%%%%%%%%%%%%%%%%%%%%%%%%%%%%
\subsection{Very Sparse Random Graphs}\label{sec:very_sparse_graphs}

In this subsection, we investigate random graphs that are close to the threshold for connectivity but are still connected, that is, we assume that $pn \le \omega \log n$ and $pn \ge \log n + \omega$ for some $\omega = \omega(n) \to \infty$ as $n \to \infty$. 

First, we will show that at time $T_3 = T_3(n) = 2 n \log n$, every node announced its opinion at least once, and at most $n \omega / \log n = o(n)$ nodes have opinion~0.

\begin{proposition}\label{prop:end_of_phase3}
Let $\omega=\omega(n) = o(\log n)$ be any function that tends to infinity (sufficiently slowly) as $n \to \infty$. Suppose that $p = p(n) \le \omega \log n / n$ and $p \ge (\log n + \omega)/n$. Set $T_3 = T_3(n) = 2 n \log n$ and $s = s(n) = n \omega / \log n$. Then, a.a.s.\ all nodes announced their opinion at time $T_3$, and at most $s$ of them have opinion~0.
\end{proposition}
\begin{proof}
Fix $T_2 = T_2(n) = n / \omega$. It follows from Proposition~\ref{prop:end_of_phase2} that a.a.s.\ $Y_1^{T_2} \ge (1/2+\delta/2) T_2$, $Y_0^{T_2} \le (1/2-\delta/2) T_2$, and trivially $Y_1^{T_2} + Y_0^{T_2} \le T_2$.  Let $V_i$ ($i \in \{0, 1\}$) be the set of nodes with opinion $i$ at time $T_2$. Since we aim for a statement that holds a.a.s., we may assume that the above inequalities are satisfied at time $T_2$ and continue the process from there. In fact, as explained in Subsection~\ref{sec:coupling}, we may assume that $|V_1| = Y_1^{T_2} = (1/2+\delta/2) T_2$ and $|V_0| = Y_0^{T_2} = (1/2-\delta/2) T_2$.

Note that during the next $T_3 - T_2 \sim 2 n \log n$ rounds, a.a.s.\ all nodes announce their opinion at least once. Indeed, the coupon collector problem is well understood and it is known that a.a.s.\ it happens after $(1+o(1)) n \log n$ rounds. We consider the nodes that announce opinion~0 at some point during this phase. In particular, let $v_1$ be the first node that announced opinion~0 during this phase, let $v_2 \neq v_1$ be the second such node, etc. For a contradiction, suppose that at time $T_3$ there are more than $s$ nodes with opinion~0. It means that the sequence we just constructed consists of more than $s$ nodes; let $S = \{v_1, \ldots, v_s\}$ be {the set of} the first $s$ nodes in this sequence. Note that all neighbours of $v_i$ in $V_1 \setminus \{v_1, \ldots, v_{i-1}\} \supseteq V_1 \setminus S$ had opinion 1 when $v_i$ announced opinion~0. On the other hand, no neighbour of $v_i$ outside of $V_0 \cup \{v_1, \ldots, v_{i-1}\} \subseteq V_0 \cup S$ had opinion~0 at that point. It follows that for any $v_i \in S$,
\begin{equation}
|N(v_i) \cap (V_1 \setminus S)| {\le} |N(v_i) \cap (V_0 \cup S)|. \label{eq:round3_prop}
\end{equation}
In fact, we will relax this property and conclude that for any $v_i \in S$, at least one of the following three properties holds:
\begin{eqnarray}
\text{ Property~(a):} && |N(v_i) \cap (V_1 \setminus S)| \le (1/2 + \delta/4) T_2 p \label{eq:round3_prop_a} \\
\text{ Property~(b):} && |N(v_i) \cap (V_0 \setminus S)| \ge (1/2 - \delta/4) T_2 p \label{eq:round3_prop_b} \\
\text{ Property~(c):} && |N(v_i) \cap S| \ge (\delta/2) T_2 p. \label{eq:round3_prop_c}
\end{eqnarray}
(Indeed, if none of properties~(\ref{eq:round3_prop_a})--(\ref{eq:round3_prop_c}) holds, then~(\ref{eq:round3_prop}) does not hold.) We partition the set $S$ into $S = S_a \cup S_b \cup S_c$: nodes in $S_x$ satisfy Property~($x$). We will show that a.a.s.\ in $\G(n,p)$ there are no sets  $V_0, V_1$, $S$, and partition $S = S_a \cup S_b \cup S_c$ such that Properties~(a)--(c) hold. This will finish the proof of the theorem. 

Let us fix $V_1 \subseteq V$ with $|V_1| = (1/2+\delta/2) T_2$, $V_0 \subseteq V \setminus V_1$ with $|V_0| = (1/2-\delta/2) T_2$, $S \subseteq V$ with $|S| = s = n \omega / \log n$, and partition $S = S_a \cup S_b \cup S_c$. (Note that these are arbitrary sets and we completely ignore the opinion dynamics process here). For any node $v_i \in S_a$, $|N(v_i) \cap (V_1 \setminus S)|$ is a binomial random variable with expectation $|V_1 \setminus S| p \sim (1/2+\delta/2) T_2 p$. (Note that $s = o(T_2)$.)  It follows from Chernoff's bound~(\ref{chern}) that $v_i \in S_a$ satisfies Property~(a) with probability at most $\exp(-\Theta(T_2p))$. Similarly, Chernoff's bound~(\ref{chern1}) implies that $v_i \in S_b$ satisfies Property~(b) with probability at most $\exp(-\Theta(T_2p))$. More importantly, the events associated with different nodes $v_i \in S_a \cup S_b$ are independent. Unfortunately, this is not the case for events associated with nodes $v_i \in S_c$. To deal with them, we need to consider all of them together. There are $\binom{|S_c|}{2} + |S_c|(s-|S_c|) \le |S_c|s$ pairs of nodes from $S$ such that at least one of them is in $S_c$. In order for nodes in $S_c$ to satisfy Property~(c), at least $(|S_c|/2) (\delta/2) T_2 p$ of such pairs must generate an edge in $\G(n,p)$. Since the expected number of edges is at most $|S_c|sp = o(|S_c|T_2p)$, by Chernoff's bound~(\ref{chern1}) we get that it happens with probability at most $\exp(-\Theta(|S_c|T_2p))$.

Note that by the union bound the probability that there exist sets  $V_0, V_1$, $S$, and partition $S = S_a \cup S_b \cup S_c$ such that Properties~(a)--(c) hold can be upper bounded by
\begin{align*}
\binom{n}{|V_1|} \binom{n-|V_1|}{|V_0|} & \binom{n}{s} (2^s)^2 \exp \left( -\Theta \Big( (|S_a|+|S_b|+|S_c|) T_2p \Big) \right) \\
&\le \binom{n}{T_2}^2 \binom{n}{s} 2^{2s} \exp \left( -\Theta ( s T_2 p ) \right) \\
&\le \left( \frac{en}{n/\omega} \right)^{2n/\omega} \left( \frac {en}{s} \right)^s 2^{2s} \exp \left( -\Theta ( s T_2 p ) \right) \\
&\le \exp \left( \bigo \left( \frac {2n \log \omega}{\omega} + s \log \log n + s \right) - \Omega \left( \frac {n \omega}{\log n} \cdot \frac {n}{\omega} \cdot \frac {\log n}{n} \right) \right) \\
&\le \exp \left( \bigo \left( \frac {2n \log \omega}{\omega} \right) - \Omega \left( n \right) \right) = o(1),
\end{align*}
which finishes the proof of the proposition.
\end{proof}

We will call a node $v$ to be of \emph{small degree}, if its degree is at most $k = 5 \log n / (\log \log n)^{1/2}$. Nodes of degree larger than $k$ will be called of \emph{large degree}. Before we continue investigating the process, we need to show a well-known fact that small degree nodes are not too close to each other.

\begin{lemma}\label{lem:small_degree}
Let $\omega=\omega(n) = o(\log \log n)$ be any function that tends to infinity (sufficiently slowly) as $n \to \infty$. Suppose that $p = p(n) \le \omega \log n / n$ and $p \ge (\log n + \omega)/n$. Then, the following property holds a.a.s.\ in $\G(n,p)$: any two small degree nodes are at distance at least 3 from each other.
\end{lemma}
\begin{proof}
Since $np > \log n$ and $k = o(\log n)$, $\binom{n}{i} p^i$ is an increasing sequence for $0 \le i \le k$ and so
\[
\Prob( \deg(v) \le k ) \le \sum_{i=0 }^k \binom{n}{i}p^i (1-p)^{n-i} \le (k+1) \binom{n}{k} p^k (1-p)^{n-k}.
\]
Using the fact that $\binom{n}{k} \le (en/k)^k$ for any integers $1 \le k \le n$ and the fact that $1-x \le \exp(-x)$ for any real number $x$, we obtain the following upper bound on the probability that a node $v$ has small degree:
\begin{eqnarray*}
\Prob( \deg(v) \le k ) 
&\le& (k+1) \binom{n}{k} p^k (1-p)^{n-k} \\
&\le& (k+1) \left( \frac {en}{k} \cdot \frac {\omega \log n}{n} \right)^k \exp \Big( -p(n-k) \Big) \\
&\le& (k+1) \left( \frac {e \omega \log n}{k} \right)^k \exp \Big( -pn+pk \Big).
\end{eqnarray*}
Recall that $k = 5 \log n / (\log \log n)^{1/2}$ and $\omega=o(\log \log n)$, so $e \omega \log n / k \le \omega (\log \log n)^{1/2} \le (\log \log n)^{3/2}$. Using this and the fact that $pn \ge \log n + \omega$ and $pk \le k \omega \log n / n = o( \log^3 n / n) = o(1)$, we get that
\begin{eqnarray*}
\Prob( \deg(v) \le k ) 
&\le& (k+1) \left( (\log \log n)^{3/2} \right)^k \exp \Big( -\log n -\omega + o(1) \Big) \\
&\le& \exp \left( \log \log n + \frac {5 \log n}{(\log \log n)^{1/2}} \cdot \frac 32 \ \log \log \log n - \log n \right) \\
&\le& \exp \Big( o(1) \log n - \log n \Big) \\
&=& n^{-1+o(1)}.
\end{eqnarray*}
Hence, we expect $n^{o(1)}$ small degree nodes and so a.a.s.\ we have only $n^{o(1)}$ of them. More importantly, using similar computations one can show that the expected number of small degree nodes that are adjacent to each other is equal to
$$
\binom{n}{2} \cdot p \cdot \Big( n^{-1+o(1)} \Big)^2 = n^{-1+o(1)} = o(1).
$$
Similarly, the expected number of pairs of small degree nodes that are at distance two from each other is equal to 
$$
\binom{n}{2} \cdot n \cdot p^2 \cdot \Big( n^{-1+o(1)} \Big)^2 = n^{-1+o(1)} = o(1).
$$
Hence, a.a.s.\ any two nodes of small degree are at distance at least three from each other, and the proof of the lemma is finished. 
\end{proof}

Our next observation is that the number of large degree nodes that have opinion~0 is decreasing. 

\begin{proposition}\label{prop:end_of_phases}
Let $\omega=\omega(n) = o(\log \log n)$ be any function that tends to infinity (sufficiently slowly) as $n \to \infty$. Suppose that $p = p(n) \le \omega \log n / n$ and $p \ge (\log n + \omega)/n$. Then, the following property holds a.a.s.\ for all phases. 

Suppose that at the beginning of a phase, $\hat{s} = (n \omega / \log n) \cdot (\log \log n)^{-(i-1)/4}$ large degree nodes have opinion~0 for some $i \in \N$. Then, after $2 n \log n$ rounds all nodes announced their opinion at least once more, and at most $u = \hat{s} / (\log \log n)^{1/4} = (n \omega / \log n) \cdot (\log \log n)^{-i/4}$ large degree nodes have opinion~0.
\end{proposition}
\begin{proof}
First, note that the expected number of nodes that were not selected in $2 n \log n$ rounds is
$$
n \left( 1 - \frac 1n \right)^{2n \log n} \le n \exp (-2 \log n) = 1/n,
$$
so with probability $1 - \bigo( 1/\log n)$ all of them are selected at least once in any phase consisting of $2 n \log n$ rounds. Since we will iteratively apply the argument for $\bigo( \log n / \log \log \log n) = o(\log n)$ phases, all of them have the desired property a.a.s.

Since we aim for a statement that holds a.a.s., we may assume that the graph satisfies property stated in Lemma~\ref{lem:small_degree}. For a contradiction, suppose that some phase fails, that is, at the beginning of this phase $\hat{s}$ large degree nodes have opinion~0, and at the end of this phase more than $u = \hat{s} / (\log \log n)^{1/4}$ large degree nodes have opinion~0. As in the proof of Proposition~\ref{prop:end_of_phase3}, we consider a sequence of distinct nodes, $v_1, v_2, \ldots$, in which large degree nodes announce opinion~0: $v_1$ announced opinion~0 first, then $v_2 \neq v_1$, etc. Let $U = \{v_1, \ldots, v_u\}$ be {the set of} the first $u$ nodes in this sequence and let $S$ be the set of large degree nodes that have opinion~0 at the beginning of this phase. Recall that each large degree node has degree at least $k = 5 (\log n)(\log \log n)^{-1/2}$ and at most one neighbour of small degree (Lemma~\ref{lem:small_degree}). Small degree nodes may (or may not) have opinion~0 but no large degree node outside of $S \cup U$ has opinion~0 at the time node $v_i$ announced opinion~0. We conclude that for all $i \in [u]$, $v_i$ has at least $k/2-1 \ge 2 (\log n)(\log \log n)^{-1/2}$ neighbours in $S \cup U$. 

We say that set $U$ satisfies Property~(a) if the following holds:
\begin{center}
Property~(a): at least $u/2$ nodes in $U$ have at least $(\log n)(\log \log n)^{-1/2}$ neighbours in $U$.
\end{center}
If $U$ does not satisfy Property~(a), then less than $u/2$ of nodes in $U$ have at least $(\log n)(\log \log n)^{-1/2}$ neighbours in $U$, which implies that set $U$ (together with $S$) satisfies the following property:
\begin{center}
Property~(b): at least $u/2$ nodes in $U$ have at least $(\log n)(\log \log n)^{-1/2}$ neighbours in $S \setminus U$.
\end{center}
We will deal with each property independently and show that it is not present in $\G(n,p)$ with the desired probability. 

If Property~(a) is satisfied for some set $U$ of size $u$, then $U$ induces at least $u (\log n)(\log \log n)^{-1/2} / 4$ edges. Hence, the probability that some set of size $u$ has this property is at most 
\begin{align*}
\binom {n}{u} & \binom { \binom{u}{2} }{u (\log n)(\log \log n)^{-1/2} / 4} p^{u (\log n)(\log \log n)^{-1/2} / 4} \\
& \le \binom {n}{u} \left( \frac {e u^2 / 2}{u (\log n)(\log \log n)^{-1/2} / 4} \cdot \frac {\omega \log n}{n} \right)^{u (\log n)(\log \log n)^{-1/2} / 4} \\
& \le \binom {n}{u} \left( \frac {2 e \omega u (\log \log n)^{1/2}} {n} \right)^{u (\log n)(\log \log n)^{-1/2} / 4} \\
& \le \binom {n}{u} \left( \frac {2 e \omega \hat{s} (\log \log n)^{1/4}} {n} \right)^{u (\log n)(\log \log n)^{-1/2} / 4} \\
& \le n^u \left( \frac {2 e \omega^2 (\log \log n)^{1/4}} {\log n} \right)^{u (\log n)(\log \log n)^{-1/2} / 4} \\
& \le \exp \left( u \log n - \frac {u (\log n)}{4 (\log \log n)^{1/2}} \cdot (1+o(1)) \log \log n \right) \\
&=  \bigo( 1/ \log n).
\end{align*}

If Property~(b) is satisfied for some set $U$ of size $u$ and some set $S$ of size $\hat{s}$, then there exists a subset $U' \subseteq U$ of size $u/2$ such that each $v_i \in U'$ has at least $(\log n)(\log \log n)^{-1/2}$ neighbours in $S \setminus U$. The probability that a given $v_i \in U'$ has this property is at most
\begin{align*}
\binom{\hat{s}}{(\log n)(\log \log n)^{-1/2}} & p^{(\log n)(\log \log n)^{-1/2}} \\
& \le \left( \frac {e\hat{s}}{(\log n)(\log \log n)^{-1/2}} \cdot \frac {\omega \log n}{n} \right)^{(\log n)(\log \log n)^{-1/2}} \\
& \le \left( \frac {e \omega^2 (\log \log n)^{1/2}}{(\log n)} \right)^{(\log n)(\log \log n)^{-1/2}} \\
& \le \exp \left( - (\log n)(\log \log n)^{-1/2} \cdot (1+o(1)) \log \log n \right) \\
& = \exp \left( - (1+o(1)) (\log n)(\log \log n)^{1/2} \right).
\end{align*}
Moreover, the events associated with different $v_i \in U'$ are independent. Hence, by the union bound, the probability that there exist a pair of sets $U, S$, and a partition $U = U' \cup (U \setminus U')$ can be upper bounded by
\begin{align*}
\binom{n}{\hat{s}} \binom{n}{u} & 2^u \exp \left( - (1+o(1)) (\log n)(\log \log n)^{1/2} \right) \\
& \le \exp \left(\hat{s} \log n + u \log n + u - (1+o(1)) (\log n)(\log \log n)^{1/2} \cdot (u/2) \right) \\
& = \exp \left( (1+o(1)) \hat{s} \log n - (1+o(1)) (\log n)(\log \log n)^{1/4} \cdot (\hat{s}/2) \right) \\
&=  \bigo( 1/ \log n).
\end{align*}
This finishes the proof of the theorem as the argument has to be (iteratively) applied only for $\bigo( \log n / \log \log \log n) = o(\log n)$ phases.
\end{proof}

Finally, we are ready to show that all nodes eventually converge to opinion~1.

\begin{proof}[Proof of Theorem~\ref{thm:very_sparse_graphs}]
The proof is an easy consequence of Propositions~\ref{prop:end_of_phase3},~\ref{prop:end_of_phases}, and Lemma~\ref{lem:small_degree}. Indeed, a.a.s.\ at time $T_3 = T_3(n) = 2n\log n$, all but at most $s = s(n) = n \omega / \log n$ nodes have opinion~1 (Proposition~\ref{prop:end_of_phase3}). Most of them are of large degree but some of them may be of small degree. By Proposition~\ref{prop:end_of_phases}, the number of large degree nodes that have opinion~0 decreases: a.a.s.\ at time $2 n \log n \cdot \bigo( \log n / \log \log n ) = \bigo( n (\log n)^2 / (\log \log n) )$ no large degree node has opinion~0. There could possibly be still some nodes of small degree that have opinion~0 but everyone converges to opinion~1 after an additional  $\bigo( n \log n)$ rounds. Indeed, every node is selected at least once during that time period a.a.s. Large degree nodes have many neighbours but at most one neighbours of small degree (Lemma~\ref{lem:small_degree}). So they will not change their opinion and stay with opinion~1. On the other hand, by the same lemma, no small degree node has a neighbour of small degree. Hence, such nodes will switch to opinion~1 once they are selected again. This finishes the proof of the theorem.
\end{proof}

%%%%%%%%%%%%%%%%%%%%%%%%%%%%%%%%
\section{Dense Random Graphs}\label{sec:dense_graphs}

In this section, we prove that for dense graphs (that is, when $p \in (0,1]$ is a constant) it is not true that all nodes converge to the correct opinion a.a.s. On the contrary, there maybe an information cascade where all the nodes converge to the wrong opinion with constant probability.

\begin{proof}[Proof of Theorem~\ref{thm:dense_graphs}]
Fix any $p \in (0,1)$. We will consider the case $p=1$ (easy case) at the end of the proof.

Trivially, the first node announces its private belief, that is, it announces opinion~1 with probability $1/2+\delta$; otherwise, it announces~0. Since nodes are selected by the process (``coupon collector'') independently of the graph, we may postpone exposing edges of the random graph till the first time a node is selected. Each time this happens, we expose edges from $v^t$ to all nodes that already announced their opinion. If every single time at least one edge is present, then all nodes are going to announce the opinion of the very first node. It follows that 
$$
p_1 \ge (1/2+\delta) \prod_{i=1}^n \Big( 1 - (1-p)^i \Big).
$$
It is easy to see that for any $x \in [0,1-p]$,
$$
f(x) = 1-x \ge \exp \left( - \frac {\log(1/p)}{1-p} x \right) = g(x).
$$ 
(Note that $f(0) = g(0)$, $f(1-p)=g(1-p)$, and $g(x)$ is convex.) Hence,
\begin{eqnarray*}
p_1 &\ge& (1/2+\delta) \exp \Big( - \frac {\log(1/p)}{1-p} \sum_{i=1}^n (1-p)^i \Big) \\
&\ge& (1/2+\delta) \exp \Big( - \log(1/p) \sum_{i=0}^{\infty} (1-p)^i \Big)\\
&=& (1/2+\delta) \exp \Big( - \log(1/p) (1/p) \Big) =(1/2+\delta)p^{1/p}.
\end{eqnarray*}
The same argument works for $p_0$ with the only difference that the probability of the first node announcing~1 ($1/2+\delta$) needs to be replaced with the probability of announcing~0 ($1/2-\delta$). 

Finally, note that if $p=1$, then the graph is (deterministically) the complete graph and (again, deterministically) all nodes are going to adopt the opinion of the very first node. Thus, we immediately get $p_1 = 1/2 + \delta$ and $p_0 = 1/2 - \delta$ {(which matches the general formula that works for $p \in (0,1]$)}. This finishes the proof of the theorem.
\end{proof}

\bibliography{ref}

%%%%%%%%%%%%%%%%%%%%%%%%%%%%%%%%
\appendix
\section{Missing Proofs}\label{sec:appendix}
%%%%%%%%%%%%%%%%%%%%%%%%%%%%%%%%

\begin{proof}[Proof of Lemma~\ref{lem:binomial}]
Let us first consider any odd value of $k \ge 3$. We get that 
\begin{eqnarray*}
q_k &=& \sum_{i \ge (k+1)/2} \binom{k}{i} (1/2+\delta/2)^i (1/2-\delta/2)^{k-i} \\
&\ge& \frac {1/2+\delta/2}{1/2-\delta/2} \cdot \binom{k}{(k+1)/2} (1/2-\delta/2)^{(k+1)/2} (1/2+\delta/2)^{(k-1)/2} \\
&& + \left( \frac {1/2+\delta/2}{1/2-\delta/2} \right)^3 \sum_{i \ge (k+3)/2} \binom{k}{i} (1/2-\delta/2)^i (1/2+\delta/2)^{k-i} \\
&=& \frac {1/2+\delta/2}{1/2-\delta/2} \cdot A + \left( \frac {1/2+\delta/2}{1/2-\delta/2} \right)^3 \cdot B,
\end{eqnarray*}
where
\begin{eqnarray*}
A &=& \binom{k}{(k+1)/2} (1/2-\delta/2)^{(k+1)/2} (1/2+\delta/2)^{(k-1)/2}  \\
B &=& \sum_{i \ge (k+3)/2} \binom{k}{i} (1/2-\delta/2)^i (1/2+\delta/2)^{k-i}.
\end{eqnarray*}
Note that $A+B = 1-q_k \le 1/2$. More importantly, if $q_k \ge 1/2 + 51 \delta / 100$, then the desired property holds and there is nothing to prove. Hence, we may assume that $1-q_k \ge 1/2 - 51\delta / 100 \ge 449/1000$. It follows that
\begin{eqnarray*}
A &=& \binom{k}{(k+1)/2} \Big( (1/2-\delta/2) (1/2+\delta/2) \Big)^{(k-1)/2} (1/2-\delta/2) \\
&=& \binom{k}{(k+1)/2} \Big( (1/4-\delta^2/4) \Big)^{(k-1)/2} (1/2-\delta/2) \\
&\le& \binom{k}{(k+1)/2} (1/2)^k \le \frac {3}{8} \le \frac {375}{449} (A+B) \le \frac {9}{10} (A+B).
\end{eqnarray*}
As a consequence,
\begin{eqnarray*}
q_k &\ge& \frac {1/2+\delta/2}{1/2-\delta/2} \cdot A + \left( \frac {1/2+\delta/2}{1/2-\delta/2} \right)^3 \cdot B \\
&\ge& \left( \frac {1/2+\delta/2}{1/2-\delta/2} \cdot \frac {9}{10} + \left( \frac {1/2+\delta/2}{1/2-\delta/2} \right)^3 \cdot \frac {1}{10} \right) (A+B),
\end{eqnarray*}
where the last inequality follows from the fact that the coefficient {in front} of $A$ is smaller than {the one in front of} $B$, so a linear combination of the coefficient's are minimized when $A$ is the largest it can {possibly} be (which is $9/10 (A+B)$). Moreover, since $A+B = 1-q_k$ we have,
\begin{eqnarray*}
\frac {q_k}{1-q_k} &\ge& \frac {1+\delta}{1-\delta} \cdot \frac {9}{10} + \left( \frac {1+\delta}{1-\delta} \right)^3 \cdot \frac {1}{10} \\
&=& (1+\delta)(1+\delta+\bigo(\delta^2)) \cdot \frac {9}{10} + (1+3\delta+\bigo(\delta^2))(1+3\delta+\bigo(\delta^2)) \cdot \frac {1}{10} \\
&=& (1+2\delta+\bigo(\delta^2)) \cdot \frac {9}{10} + (1+6\delta+\bigo(\delta^2)) \cdot \frac {1}{10} \\
&=& (1+2\delta+\bigo(\delta^2)) \cdot \frac {9}{10} + (1+6\delta+\bigo(\delta^2)) \cdot \frac {1}{10} \\
&=& 1 + \frac {12}{5} \delta + \bigo(\delta^2) = \frac {1 + 6\delta/5}{1-6\delta/5} + \bigo(\delta^2) = \frac {1/2 + 3\delta/5}{1/2-3\delta/5} + \bigo(\delta^2).
\end{eqnarray*}
Clearly, 
$$
\frac {q_k}{1-q_k} \ge \frac {1/2 + 51\delta/100}{1/2-51\delta/100}
$$
for sufficiently small $\delta$ but one can show it holds for $\delta \in (0,1/10]$. This implies $q_k \ge {1/2+51\delta/100}$ for odd values of $k\ge 3$.

Let us now consider any even value of $k \ge 2$. This time we get that 
\begin{eqnarray*}
q_k &=& \sum_{i \ge k/2+1} \binom{k}{i} (1/2+\delta/2)^i (1/2-\delta/2)^{k-i} + \binom{k}{k/2} (1/2+\delta/2)^{k/2} (1/2-\delta/2)^{k/2} (1/2+\delta) \\
&\ge& \left( \frac {1/2+\delta/2}{1/2-\delta/2} \right)^2 \sum_{i \ge k/2+1} \binom{k}{i} (1/2-\delta/2)^i (1/2+\delta/2)^{k-i} \\
&& + \frac {1/2+\delta}{1/2-\delta} \cdot \binom{k}{k/2} (1/2+\delta/2)^{k/2} (1/2-\delta/2)^{k/2} (1/2-\delta) \\
&=& \left( \frac {1+\delta}{1-\delta} \right)^2 \cdot A + \frac {1/2+\delta}{1/2-\delta} \cdot B, 
\end{eqnarray*}
where 
\begin{eqnarray*}
A &=& \sum_{i \ge k/2+1} \binom{k}{i} (1/2-\delta/2)^i (1/2+\delta/2)^{k-i} \\
B &=& \binom{k}{k/2} (1/2+\delta/2)^{k/2} (1/2-\delta/2)^{k/2} (1/2-\delta),
\end{eqnarray*}
and $A+B = 1-q_k$. Since
$$
\left( \frac {1+\delta}{1-\delta} \right)^2 = \frac {1+2\delta + \delta^2}{1-2\delta+\delta^2} \ge \frac {1+ 19\delta/10}{1-19\delta/10} = \frac {1/2+ 19\delta/20}{1/2-19\delta/20} \ge \frac {1/2+ 51\delta/100}{1/2-51\delta/100}
$$
and, trivially,
$$
\frac {1/2+\delta}{1/2-\delta} \ge \frac {1/2+51\delta/100}{1/2-51\delta/100},
$$
we get that
$$
q_k \ge  \frac {1/2+51\delta/100}{1/2-51\delta/100} (1-q_k),
$$
which implies $q_k \ge  {1/2+51\delta/100}$ for even values of $k$ too.
\end{proof}

\end{document}